\documentclass[reprint, amsmath,amssymb, aps, pra]{revtex4-2}
\usepackage[english]{babel}

\usepackage{bbm,braket,mathrsfs,amsmath,amssymb,color,amsthm,mathtools,enumitem,ae,aecompl,bm,thm-restate,hyperref}
\usepackage{graphicx}% Include figure files
\usepackage{dcolumn}% Align table columns on decimal point
\usepackage[capitalize]{cleveref}
\usepackage{silence}

\allowdisplaybreaks

\newcommand{\ketbra}[1]{\ket{#1}\!\bra{#1}}
\newcommand{\trace}[1]{\mathrm{Tr}\!\left ( {#1}\right )}
\newcommand{\absval}[1]{\left| {#1} \right|}
\newcommand{\e}[1]{\exp\left\{#1\right\}}
\newtheorem{theorem}{Theorem}

\newtheorem{lemma}[theorem]{Lemma}

\begin{document}
\title{Quantum States with Maximal Magic}
\author{Gianluca Cuffaro}
\author{Christopher A. Fuchs}
\affiliation{Department of Physics, University of Massachusetts Boston \\ 100 Morrissey Boulevard, Boston MA 02125, USA}

\begin{abstract}
Finding ways to quantify magic is an important problem in quantum information theory. Recently Leone, Oliviero and Hamma introduced a class of magic measures for qubits, the stabilizer entropies of order $\alpha$, to aid in studying nonstabilizer resource theory. This suggests a way to search for those states that are as distinct as possible from the stabilizer states. In this work we explore the problem in any finite dimension $d$ and characterize the states that saturate an upper bound on stabilizer entropies of order $\alpha\geq2$. Particularly, we show that if a Weyl-Heisenberg (WH) covariant Symmetric Informationally Complete (SIC) quantum measurement exists, its states \emph{uniquely} maximize the stabilizer entropies by saturating the bound.  No other states can reach so high. This result is surprising, as the initial motivation for studying SICs was a purely quantum-foundational concern in QBism.  Yet our result may have implications for quantum computation at a practical level, as it demonstrates that this notion of maximal magic inherits all the difficulties of the 25-year-old SIC existence problem, along with the deep questions in number theory associated with it.  
\end{abstract}

\maketitle
\section{Introduction}
The notion of a ``magic state'' for quantum computation was introduced by Bravyi and Kitaev~\cite{kitaev} in 2004 as a means of getting around the Gottesman-Knill theorem~\cite{gottesman1998heisenbergrepresentationquantumcomputers}. This theorem states that any quantum computation with Clifford-group quantum gates and stabilizer-state inputs can be simulated efficiently via a classical computation.  In other words, the theorem identifies a fragment of quantum theory that is essentially classical. However, if these computations are supplemented with the injection of ``magic states'' everything changes---in particular, one regains universal quantum computation, which cannot be simulated efficiently with a classical computer.

Thus, it is not surprising that a lot of effort has been put into characterizing this class of states. In this context, resource theory~\cite{QRT} may provide a crucial tool for approaching the problem, with the stabilizer states playing the role of a free resource. Yet, finding a well-behaved measure with this idea as its basis is still no easy task. Although a number of such measures have been introduced---for instance various notions of Wigner-function ``negativity'' \cite{Veitch_2012, nsf30, PhysRevA.102.032221} and mana~\cite{Veitch_2014,PhysRevLett.109.230503,warmuz2024magicmonotonefaithfuldetection,PhysRevLett.118.090501,PhysRevLett.115.070501}---usually they either have a limited range of application to specific dimensions or require taxing optimization procedures for their calculation. This makes them less viable from a practical standpoint. Recently however, Leone et al.\ \cite{Leone_2022} provided a class of measures of magic for qubit systems that are well-defined and easy to compute in real-world applications~\cite{PhysRevLett.132.240602sss,PhysRevLett.133.010601vvvvv,PhysRevLett.131.180401aaaa,oliviero_measuring_2022}. These are the so-called stabilizer entropies of order $\alpha$ (based on the classic R\'enyi entropies of order $\alpha$~\cite{1975measures}). 

With their aid, one can easily characterize the free resources in this theory, namely the stabilizer states. On the other hand, it has been harder to find the maximum of these measures for the purpose of characterizing the states most distinct from stabilizer states. In the literature~\cite{andersson_states_2015, wang_stabilizer_2023} the prime suspects for that class have been the states constituting the so-called Weyl-Heisenberg covariant Symmetric Informationally Complete POVMs~\cite{Renes_2004,Zauner:phd}, or WH-SICs for brevity.  However, until now it has not been clear whether there might also be other states with maximal magic.

For any $d$-dimensional Hilbert space, the SICs (i.e., those with or without the Weyl-Heisenberg condition) are very special POVMs consisting of $d^2$ rank-1 projectors with equal pairwise Hilbert-Schmidt inner products. Unfortunately, these objects are as hard to analyze mathematically as they are simple to define. Clearly, they have a very simple geometrical interpretation since they always distribute in quantum-state space on the vertices of a regular simplex. However, the question of whether they exist in all finite dimensions is a 25-year-old unsolved problem in quantum information theory~\cite{axioms6030021}. To date, they have been proven to exist in 117 different dimensions, and there is strong numerical evidence of their existence (to thousands of digits of accuracy) in about 100 dimensions more than that~\cite{grasslprivate}. However a full proof of their existence in every $d$ may be as elusive as it ever was, as it has recently become understood that the question is deeply connected to one of Hilbert's famous turn-of-the-century problems---problem 12 to be exact, to do with algebraic number theory~\cite{nsf17}. Yet, their interest for physics is already firmly established by the overwhelming number of problems \emph{they would solve if they were to exist}.  These range from quantum foundations where they would give a means of expressing the Born rule as a minimal deviation from the classical Law of Total Probability~\cite{fuchs2023qbismnext, DeBrota2020}, to Lie and Jordan algebraic questions in mathematical physics~\cite{nsf21}, to good detectors of quantum entanglement~\cite{nsf7}, to elements in novel quantum key distribution schemes~\cite{nsf8,fuchs2003squeezingquantuminformationclassical}, to components in optimal device-independent quantum random number generation~\cite{nsf9,nsf10}, to dimension witnessing~\cite{nsf11}, and optimal quantum-state tomography~\cite{nsf12,nsf13,nsf14}.

In our work, we consider an extension of the stabilizer entropies of order $\alpha$, as in Ref.\ \cite{wang_stabilizer_2023}, to generic $d$-dimensional systems, and prove that for any $\alpha\geq2$, a state can achieve a certain upper bound on magic \emph{if and only if} it belongs to a WH-SIC\@. This result establishes that if a WH-SIC exists in all finite dimensions $d$, then its states will always uniquely possess the maximum allowed magic.  That is to say, the WH-SICs are a kind of (quantum) antithesis to the (classical) stabilizer states.  This adds to the growing list of reasons for why the SICs may be ultimately important structures within quantum information theory despite the difficulty (or perhaps precisely \emph{because of} the difficulty) of establishing their existence.  It further establishes that the concept of magic, at least as quantified by the stabilizer entropies, is deeply number theoretic in character.

The paper is organized as follows. In Section~\ref{II} we present basic notions from the stabilizer formalism and introduce the stabilizer entropies as measures of magic. In Section \ref{3} we define the Weyl-Heisenberg covariant SIC-POVMs and present some of their more notable properties. In Section~\ref{4} we prove our main theorem. Finally, in Section~\ref{5} we discuss the implications of our result and present some future perspectives.

\section{Stabilizer formalism and stabilizer entropy} \label{II}
\subsection{Weyl-Heisenberg group}
For any natural number $d$, let $\mathbb{Z}_d=\{0,\dots,n-1\}$ be the set of congruence classes modulo $d$. We define the computational basis of a $d$-dimensional Hilbert space $\mathcal{H}_d$ as the orthonormal set $\{\ket{k}\,|\,k\in\mathbb{Z}_{d}\}$. We then define the shift and the clock operators respectively as:
\begin{equation}
    X\ket{k}\equiv\ket{k+1}\quad\quad,\quad\quad Z\ket{k}\equiv\omega^k\ket{k},
\end{equation}
with $\omega=e^{\frac{2\pi i}{d}}$ and the arithmetic operations being defined modulo $d$.

The shift and clock operators satisfy the following properties:
\begin{equation}
    X^d=Z^d=\mathbb{I}\quad\quad,\quad\quad X^kZ^l=\omega^{-kl}Z^lX^k.
\end{equation}
Using them, we define the displacement operators as
\begin{equation}
    D_\textbf{a}\equiv\omega^{\frac{a_1a_2}{2}}X^{a_1}Z^{a_2},\quad\textbf{a}=(a_1,a_2)\in\mathbb{Z}_d\times\mathbb{ Z}_d.
\end{equation}
Using the properties of the shift and clock operators, it is easy to show that the displacement operators satisfy a number of interesting properties: 
\begin{align}
    &D_{\textbf{a}}^\dagger=D_{-\textbf{a}},\label{dagger}\\
    &D_\textbf{a}D_\textbf{b}=\omega^{\frac{[\textbf{a,\textbf{b}}]}{2}}D_{\textbf{a}+\textbf{b}}=\omega^{[\textbf{a},\textbf{b}]}D_{\textbf{b}}D_\textbf{a},\label{HEalgebra}\\
    &\trace{D_\textbf{a}D_\textbf{b}^\dagger}=d\delta_{\textbf{ab}},\label{orth}
\end{align}
where $[\textbf{a},\textbf{b}]=a_1b_2-a_2b_1$ is the symplectic product on $\mathbb{Z}_d\times\mathbb{ Z}_d$. In particular, \cref{orth} establishes the displacement operators as an orthogonal basis of unitary operators.

The Weyl-Heisenberg (WH) group $\mathcal{W}(n)$ is generated by all the displacement operators:
\begin{equation}
    \mathcal{W}(n)\equiv\{\omega^sD_\textbf{a}\}\,,\quad\quad s\in\mathbb{Z}_n,\,\,\textbf{a}\in\mathbb{Z}_d\times\mathbb{ Z}_d.
\end{equation}
Notice that the order of the group is given by $d\times d^2$.

The same structure can be easily replicated in the case of composite systems. Consider a composite Hilbert space $\mathcal{H}_d=\mathcal{H}_{n_1}\otimes\dots\otimes \mathcal{H}_{n_k}$, with dimension $d=n_1\dots n_k$. Displacement operators acting on $\mathcal{H}_d$ are defined via the tensor product of the corresponding operators on each subsystem~\cite{gross, Veitch_2012}. The composite Weyl-Heisenberg group is therefore the group generated by the composite displacement operators:
\begin{equation}
    \mathcal{W}(d)\equiv\{\omega^sD_{\textbf{a}_1}\otimes\dots\otimes D_{\textbf{a}_k}\}
\end{equation}
which has again order $d\times d^2$.

Finally, let $\overline{\mathcal{W}}(d)=\mathcal{W}(d)/\langle\omega^s\mathbb{I}\rangle$ be the quotient group of the WH group containing only $+1$ phases. Not being interested in global phases, we will only consider the quotient group $\overline{\mathcal{W}}(d)$, and for the sake of notation we will drop the overline. Let us remark that after discarding the nontrivial phases, the order of the group becomes $d^2$.

As a final note, it is worth pointing out that the structure of the Weyl-Heisenberg group can depend strongly on the number theoretical properties of the dimension of the finite space $\mathbb{Z}_d$. Indeed, many results concerning the WH group are valid only if one considers prime power dimensions, even or odd dimensions, and finally multipartite qubit or qudit systems~\cite{bengtsson2017discretestructuresfinitehilbert, Zhu_2010, 10.1063/1.1896384,gross}. Even if some of these results will have consequences in conjunction with the ones we present in this work, our proofs will be independent of dimension. We will reserve the discussion about the implications of the different possible versions of the WH group for Section \ref{5}. For the sake of notation, we will refer to any choice of WH group as $\mathcal{W}(d)$.

\subsection{Clifford group and Stabilizer states}
Consider the unitaries $U$ such that for any $D_\textbf{a}\in\mathcal{W}(d)$
\begin{equation}
    U^\dagger D_\textbf{a}U=e^{i\gamma}D_{\textbf{a}'},
\end{equation}
for some $\textbf{a}'\in\mathbb{Z}_d\times\mathbb{Z}_d$ and some $\gamma\in\mathbb{R}$.
A straightforward calculation shows that they form a group. In particular, the group of unitaries $\mathcal{C}(d)$ with this property is called the Clifford group, and by definition it is the normalizer of the Weyl-Heisenberg group:
\begin{equation}
    \mathcal{C}(d)=\{U\in\mathcal{U}(d)\,|\,U^\dagger\mathcal{W}(d)U\subseteq\mathcal{W}(d)\}.
\end{equation}
Namely, the Clifford group is the subgroup of the unitaries that permutes (up to a phase) the displacement operators. Another crucial feature of Clifford operators is that they preserve the set of stabilizer states. 

We can define the stabilizer states by noticing that according to the Weyl-Heisenberg algebra described in \cref{HEalgebra} whenever $[\textbf{a},\textbf{b}]=0$, the corresponding displacement operators $D_\textbf{a}$ and $D_\textbf{b}$ commute. Hence, if can we find a subset $M\subset \mathbb{Z}_d\times\mathbb{Z}_d$ such that $\forall\,\textbf{a},\textbf{b}\in M$, $[\textbf{a},\textbf{b}]=0$, we can simultaneously diagonalize all the corresponding displacement operators. The maximal cardinality of such a subset $M$ is given by $\absval{M}=d$. In this case there exists a unique (up to a global phase) vector $\ket{M}$ such that
\begin{equation}\label{stab}
    D_{\textbf{a}}\ket{M}=\ket{M}\quad\quad \forall\textbf{a}\in M.
\end{equation}
The vector $\ket{M}$ is called a stabilizer state, and the abelian subgroup of $\mathcal{W}(d)$ that satisfies \cref{stab} is the stabilizing group of $\ket{M}$. It is possible to show~\cite{gross} that we can always write the rank-1 projector corresponding to $\ket{M}$ as 
\begin{equation}
    \ketbra{M}=\frac{1}{d}\sum_{\textbf{a}\in M}D_\textbf{a}.
\end{equation}
It is clear now why the Clifford unitaries send stabilizer states into other stabilizer states.

As a final remark, notice that in this work we are only interested in pure stabilizer states. It is possible to extend the definition to mixed states, e.g., by taking convex combination of pure stabilizer states~\cite{warmuz2024magicmonotonefaithfuldetection}. However, mixed stabilizer states pose new and different problems that are beyond the scope of our work.

\subsection{Stabilizer Entropy}
Consider a vector $\ket{\psi}\in\mathcal{H}_d$. Since the WH operators form an orthogonal basis for the space of linear operators on $\mathcal{H}_d$, we can expand the state $\psi\equiv\ketbra{\psi}$ in terms of displacement operators:
\begin{equation}
    \psi=\frac{1}{d}\sum_{\textbf{a}}\trace{D_\textbf{a}^\dagger\psi}D_\textbf{a}.
\end{equation}

The function $\chi_\textbf{a}(\psi)\equiv\frac{1}{d}\trace{D_\textbf{a}\psi}$ is usually referred to as the characteristic function of the state $\psi$~\cite{gross,Veitch_2012}. Since the $D_\textbf{a}$'s are not hermitian operators (except for the case of $d=2^n$), in general the characteristic function of a given state will be a complex-valued function. However, by taking the absolute value squared of $\chi_\textbf{a}(\psi)$ and appropriately rescaling it, one can verify that the quantity
\begin{equation}
    P_\textbf{a}(\psi)\equiv\frac{1}{d}\absval{\trace{D_\textbf{a}\psi}}^2\,,\quad\quad D_\textbf{a}\in \mathcal{W}(d)
\end{equation}
forms a probability vector. Hence, we can make use of the probability vector's $\alpha$-Rényi entropy to define a \textit{stabilizer entropy} of order $\alpha$~\cite{Leone_2022,wang_stabilizer_2023} according to:
\begin{equation}
    M_\alpha(\psi)\equiv\frac{1}{1-\alpha}\log{\sum_\textbf{a}P_\textbf{a}(\psi)^\alpha}-\log{d}\;,
\end{equation}
where $\alpha\ge 2$ and the $-\log d$ offset serves to make the quantity lower-bounded by zero.  (Let us remark again that the measures as defined here work only for pure states, even if it is possible to extend them to mixed states~\cite{Leone_2022}.) 

The stabilizer entropies have the following properties which establish them as good measures from a resource-theory perspective~\cite{wang_stabilizer_2023, QRT}:
\begin{itemize}
    \item[(i)] faithfulness: \begin{equation}
        M_\alpha(\psi)=0\quad \text{iff}\quad \ket{\psi}\,\, \text{is a stabilizer state};
    \end{equation}
    \item[(ii)] invariance under Clifford operations:  \begin{equation}
        \forall C\in\mathcal{C}(d)\quad\quad M_\alpha(C^\dagger\psi C)=M_\alpha(\psi);
    \end{equation}
    \item[(iii)] additivity under tensor product:
    \begin{equation}
        M_\alpha(\psi\otimes\phi)=M_\alpha(\psi)+M_\alpha(\phi);
    \end{equation}
\end{itemize}

The main focus of this paper is the characterization of the states that reach the maximum possible value of stabilizer entropy of order $\alpha\ge 2$. In particular, we prove that a quantum state can attain an upper bound on the stabilizer entropy given by
\begin{equation}\label{upb}
        M_\alpha(\psi)\leq\frac{1}{1-\alpha}\log{\frac{1+(d-1)(d+1)^{1-\alpha}}{d}}\;.
    \end{equation}
\emph{if and only if} it belongs to a WH-SIC\@.
    
\section{Symmetric Informationally Complete Measurements and t-designs} \label{3}
The most general form of measurement in quantum mechanics is represented by a positive operator valued measure (POVM)~\cite{peres_quantum_2002,Nielsen_Chuang_2010}. A POVM consists in a set of positive semi-definite operators $\{E_i\}_{i=1}^n$ such that $\sum_i E_i = \mathbb{I}$. Such operators capture the statistics of a measurement by assigning to each outcome $i$ an effect $E_i$, whose probability is given by the Born rule $p(i)=\trace{E_i\rho}$, where $\rho$ is a pure or mixed quantum-state assignment. A POVM is said to be informationally complete if its outcome probabilities $p(i)$ uniquely determine the original quantum state $\rho$. It is clear that to be informationally complete a POVM must have at least $d^2$ elements when the Hilbert space $\mathcal{H}_d$ is of dimension $d$.  In this way, its elements will span the $d^2$-dimensional space of linear operators on $\mathcal{H}_d$.

A \emph{symmetric informationally complete} set of quantum states~\cite{Renes_2004} (or, SIC hereafter) is a set of rank-1 projection operators $\{\Pi_i=\ketbra{\psi_i}\}_{i=1}^{d^2}$ such that
\begin{equation}\label{sic1}
    \trace{\Pi_i\Pi_j}=\absval{\braket{\psi_i|\psi_j}}^2=\frac{d\delta_{ij}+1}{d+1}.
\end{equation}
A consequence of this definition is that if such a set exists, the $\Pi_i$ in it must be linear independent.  Thus when rescaled to $E_i = \frac{1}{d}\Pi_i$ one obtains an informationally complete POVM.

In Ref.\ \cite{appleby_symmetric_2010} (co-authored by one of us), it was shown that the SICs can be characterized as sets of $d^2$ normalized positive semi-definite operators $\mathcal{A}=\{A_i\}$, i.e., with
\begin{equation}
\trace{A_i^2}=1\;,   
\end{equation}
that achieve a lower bound on a class of  ``orthogonality'' measures with respect to the Hilbert-Schmidt inner product. We must point out a major proviso with that result however: There was a typo in its start that propagated through the rest of the paper causing the numerical value of the bound to be incorrect. Taking that into account, the $\alpha$-orthogonality of a set for \emph{any real number} $\alpha\ge1$ is (now correctly) defined by
\begin{equation}
K_\alpha[\mathcal{A}]= \sum_{i \neq j=1}^{d^2}
\bigl(\trace{A_i A_j}\bigr)^{2\alpha}\;. \label{eq:KDef}
\end{equation}
Clearly this quantity would equal zero if the $A_i$ could be orthogonal to each other. However, with the constraint of positive semi-definiteness taken into account, a (corrected) nontrivial lower bound arises as
\begin{equation}
K_\alpha[\mathcal{A}] \ge \frac{d^2(d-1)}{(d+1)^{2\alpha-1}}\;. \label{eq:KBound}
\end{equation}
Equality is achieved if and only if the $A_i$ satisfy the conditions for a SIC\@. Notable in the derivation of Ref.\ \cite{appleby_symmetric_2010} is that the $A_i$ were not assumed to be rank-1 operators at the outset---``rank-1'ness'' had to be achieved.  In the next section, we will revisit the (corrected) derivation of Ref.\ \cite{appleby_symmetric_2010} for the special case that the $A_i$ are rank-1.

Let us remark that the quantity $K_\alpha[\mathcal{A}]$ looks superficially similar to the notion of a $t$-th order frame potential from design theory~\cite{Renes_2004}:
\begin{equation}\label{framepot}
F_t[\mathcal{V}]\equiv\sum_{j,k=1}^m\absval{\braket{\phi_j|\phi_k}}^{2t}\;,
\end{equation}
where $\mathcal{V}$ denotes sets of pure states $\mathcal{V}=\{|\phi_1\rangle, \ldots, |\phi_m\rangle\}$ with a floating cardinality $m$, and $t\ge 1$ is specifically an integer. So $m$ is not necessarily $m=d^2$, and the $t$ are integers in contrast to the general real-valued $\alpha$. This is because in design theory the concern is to find a finite set of pure states $\mathcal{V}$ such that
\begin{align}
\frac{1}{m}\sum_{i=1}^m  |\phi_i\rangle\langle \phi_i|^{\otimes t} = \int |\psi\rangle\langle\psi|^{\otimes t} d\psi\;.
\end{align}
Such a set of states can be found alternately through a $\mathcal{V}$ which achieves the global minimum of the function \cref{framepot}. When $t>2$, one can show that an $m$ allowing for a global minimum is strictly greater than $d^2$.  Note, however, that such is \emph{not} the optimization problem we are concerned with in this paper.  

Next, we say that a SIC is covariant with respect to a group $G$ if there exists at least one normalized fiducial vector $\ket{\phi}\in\mathcal{H}_d$ such that
\begin{equation}\label{Hermann}
    \absval{\bra{\phi}U_g\ket{\phi}}^2=\frac{1}{d+1},\quad\quad\forall\, U_g\ne \mathbb{I} \in G.
\end{equation}
In other words a SIC is covariant if there is a  group orbit $\mathcal{G}_\phi\equiv\{U_g\ket{\phi}\,\,|\,\,U_g\in G\}$ that satisfies~\cref{Hermann}. 

All known SICs are covariant under the single-qudit WH group, with the single exception of the Hoggar SIC in $d=8$ which covariant under the three-qubit WH group $\mathcal{W}(2)^{\otimes3}$. The question of whether there are SICs covariant under still different groups (or not covariant with respect to any group at all) is still open \cite{Zhu_2010, Samuel_2024}. 

As one might expect, an important feature of the WH-SICs is that Clifford unitaries send their fiducials into other fiducials. Indeed, if $\ket{\psi}$ is a fiducial vector for a WH-SIC then by definition
\begin{equation}\label{sic}
    \absval{\bra{\psi}D_\textbf{a}\ket{\psi}}^2=\begin{dcases}
        \quad1 \quad\quad& \mbox{if} \;\textbf{a}=0 \\ \frac{1}{d+1} \quad\quad&\,\text{otherwise.}
    \end{dcases}
\end{equation}
Let $U\in\mathcal{C}(d)$. Since the Clifford group is the normalizer of the Weyl-Heisenberg group, the action $U^\dagger D_\textbf{a}U$ will give another displacement operator up to a phase. Hence, for all elements $U\in\mathcal{C}(d)$, if $\ket{\psi}$ satisfies the SIC condition in \cref{sic}, then  $U\ket{\psi}$ will also. 

\section{Characterization of the upper bound on the stabilizer entropies} \label{4}

In this section we will finally prove our main theorem:
\begin{theorem}\label{thm2}
A state $\ket{\psi}$ will achieve our upper bound on the stabilizer entropy of any order $\alpha\ge 2$ 
    \begin{equation}
        M_\alpha(\psi)\leq\frac{1}{1-\alpha}\log{\frac{1+(d-1)(d+1)^{1-\alpha}}{d}}.
    \end{equation}
if and only if it belongs to a Weyl-Heisenberg covariant SIC\@.
\end{theorem}

We will approach the proof of this theorem by first demonstrating a couple of useful lemmas. The first establishes a lower bound on the $\alpha$-orthogonality $K_\alpha[\mathcal{V}]$, which we prove can be achieved if and only if $\mathcal{V}$ is a SIC-POVM. The second relates an $\alpha$-stabilizer entropy to $K_\alpha[\mathcal{V}]$ in the case of sets $\mathcal{V}$ that are orbits of the Weyl-Heisenberg group.

\begin{lemma}\label{lemm2}
Consider a set of normalized vectors $V\equiv\{\ket{\phi_i}\}\subset\mathcal{H}_d$ with cardinality $\absval{\mathcal{V}}=d^2$. Then,
\begin{equation}\label{FunkaBunka}
    K_\alpha[\mathcal{V}] \ge \frac{d^2(d-1)}{(d+1)^{2\alpha-1}}\;,
\end{equation}
with equality if and only if $\mathcal{V}$ is a SIC\@.
\end{lemma}
\begin{proof}
    Our proof follows closely to the argument of \cite{appleby_symmetric_2010}, while as pointed out earlier, correcting the mistakes in the original presentation. First, let us define the quantity
    \begin{equation}
        K_\alpha[\mathcal{V}]=\sum_{i\neq j}\big(\trace{\phi_i\phi_j}\!\big)^{2\alpha}\;,
    \end{equation}
where $\phi_i\equiv\ketbra{\phi_i}$ and $\ket{\phi_i}\in \mathcal{V}$.  Note that as opposed to the frame potential in \cref{framepot}, this sum runs only over the
cases where $i\ne j$. The reason for this exclusion will become apparent shortly as it will aid in deriving a lower bound via standard techniques based on the Cauchy-Schwarz inequality and convexity.

Particularly, note the following simple consequence of the Cauchy-Schwarz inequality.  Let $a_i$, $i=1,\ldots,n$, be any real numbers.  Then, 
\begin{equation}\label{34000ft}
\sum_{i=1}^n a_i^2\ge\frac{1}{n}\!\left(\sum_{i=1}^n a_i\right)^{\!2}\;,    
\end{equation}
where equality is achieved if and only if the $a_i=const\quad\forall i$.

We will first consider the case $\alpha=1$.  Let $G\equiv\sum_i\phi_i$. Then we note that 
    \begin{equation}\label{fragglerock}
        \trace{G^2}=\sum_{i,j}\trace{\phi_i\phi_j},
    \end{equation}
and using \cref{34000ft}, we have
\begin{eqnarray}
K_1[\mathcal{V}] &\ge& \frac{1}{d^4-d^2}\!\left(\sum_{i\ne j} \trace{\phi_i \phi_j}\right)^{\! 2}
\\
&=& \frac{1}{d^4-d^2}\Big(\trace{G^2}-d^2\Big)^2\;. 
\end{eqnarray}
On the other hand, if we let the eigenvalues of $G$ be denoted by $\lambda_a$, we can use \cref{34000ft} again to find
    \begin{equation}
        \trace{G^2}=\sum_{a=1}^d\lambda_a^2\geq\frac{1}{d}\!\left(\sum_{a=1}^d\lambda_a\right)^{\!2}=\frac{1}{d}\big(\trace{G}\!\big)^{2},
    \end{equation}
But
    \begin{equation}
        \trace{G}=\sum_i\trace{\phi_i}=d^2.
    \end{equation}
Putting these steps together we find that
    \begin{equation}\label{multifactoid}
        K_1[\mathcal{V}] \ge \frac{d^2(d-1)}{d+1}\;,
    \end{equation}
Equality is achieved only when both instances of the Cauchy-Schwarz achieve equality, namely,
\begin{enumerate}
\item
$\trace{\phi_i \phi_j}={\rm constant}$ for all $i\ne j$\;,
\item
$\lambda_a={\rm constant}$ for all $a$\;.
\end{enumerate}
The trace conditions taken together imply that $G$ must be proportional to the identity operator, in fact $G=d\,\mathbb{I}$. Consequently, with the aid of \cref{fragglerock},
\begin{equation}
\trace{\phi_i \phi_j}=\frac{1}{d+1} \quad \forall\, i\ne j\;.
\end{equation}
In other words, \cref{multifactoid} achieves equality if and only if $\mathcal{V}$ is a SIC.
    
Next, let us consider the cases where $\alpha>1$. Notice that the function $f(x)=x^\alpha$ is a strictly convex function. This suggests that Jensen's inequality is the next relevant thing to consider.  Namely, for any strictly convex function,
    \begin{equation}
        \sum_{k}^n f\!\left(x_k\right)\geq n f\!\left(\frac{1}{n}\sum_{i}^nx_k\right),
    \end{equation}
with equality if and only if all the $x_k$ are identical.
In the case of $K_\alpha[\mathcal{V}]$, where there are $d^4-d^2$ terms in the sum, we have 
    \begin{equation}\label{eq}
        K_\alpha[\mathcal{V}]\geq\big(d^4-d^2\big)\!\!\left(\frac{1}{d^4-d^2}K_1[\mathcal{V}]\right)^{\!\alpha}\geq \frac{d^2(d-1)}{(d+1)^{2\alpha-1}}\;.
    \end{equation}
Again equality is achieved if and only if $\mathcal{V}$ is a SIC.  This completes the proof of our lemma.
\end{proof}

    \begin{lemma}\label{lemm}
        Define the orbit of the Weyl-Heisenberg group passing through a state $\ket{\phi}$ as the set $\mathcal{D}_\phi\equiv\{D_\textbf{a}\ket{\phi}\,|\,D_\textbf{a}\in \mathcal{W}(d)\}$. Then, 
        \begin{equation}
            K_\alpha[\mathcal{D}_\phi]=d^3\e{(1-2\alpha)M_{2\alpha}(\phi)}-d^2.
        \end{equation}
    \end{lemma}
    \begin{proof}
    Let us denote $\phi_\textbf{a}\equiv D_\textbf{a}\ketbra{\phi} D_\textbf{a}^\dagger\in\mathcal{D}_\phi$. We have
    \begin{align}
        K_\alpha[\mathcal{D}_\phi]&=\sum_{\textbf{ab},\,\textbf{a}\neq\textbf{b}}\trace{\phi_\textbf{a}\phi_\textbf{b}}^{2\alpha}\\
        &=\sum_{\textbf{a}\textbf{b},\,\textbf{a}\neq\textbf{b}}\trace{D_\textbf{a}\ketbra{\phi}D_\textbf{a}^\dagger D_\textbf{b}\ketbra{\phi}D_\textbf{b}^\dagger}^{2\alpha}\\
        &=\sum_{\textbf{a}\textbf{b},\,\textbf{a}\neq\textbf{b}}\left(\bra{\phi}D_\textbf{a}^\dagger D_\textbf{b}\ket\phi\,\bra{\phi}D_\textbf{b}^\dagger D_\textbf{a}\ket\phi\right)^{2\alpha}\\
        &=\sum_{\textbf{a}\textbf{b},\,\textbf{a}\neq\textbf{b}}\absval{\bra{\phi}D_{\textbf{a}}^\dagger D_\textbf{b}\ket\phi}^{4\alpha}.
    \end{align}
    Since we are only considering rank-1 projectors, we can easily reintegrate the diagonal terms into the sum, obtaining
    \begin{equation}
        K_\alpha[\mathcal{D}_\phi]=\sum_{\textbf{a}\textbf{b}}\absval{\bra{\phi}D_{\textbf{a}}^\dagger D_\textbf{b}\ket\phi}^{4\alpha}-d^2.
    \end{equation}
    Then, using the properties in \cref{HEalgebra,dagger}, we have that
    \begin{align}
        K_\alpha[\mathcal{D}_\phi]&=\sum_{\textbf{a}\textbf{b}}\absval{\bra{\phi}D_{-\textbf{a}}D_\textbf{b}\ket{\phi}}^{4\alpha}-d^2\\&=\sum_{\textbf{a}\textbf{b}}\absval{\bra{\phi}D_{\textbf{b}-\textbf{a}}\ket{\phi}}^{4\alpha}-d^2.
    \end{align}
    Now, relabeling the indices such that $\textbf{b}-\textbf{a}\equiv\textbf{b}'$, we have
    \begin{align}
        K_\alpha[\mathcal{D}_\phi]&=\sum_{\textbf{a}\textbf{b}'}\absval{\bra{\phi}D_{\textbf{b}'}\ket{\phi}}^{4\alpha}-d^2\\
        &=d^2\sum_{\textbf{b}'}\absval{\bra{\phi}D_{\textbf{b}'}\ket{\phi}}^{4\alpha}-d^2.
    \end{align}
    Notice that we can write the $2\alpha$-th stabilizer entropy as
    \begin{equation}
        M_{2\alpha}(\phi)=\frac{1}{1-2\alpha}\log{\frac{1}{d}\sum_\textbf{a}\absval{\bra{\phi}D_\textbf{a}\ket{\phi}}^{4\alpha}}\;.
    \end{equation}
    Therefore, we have the claim
    \begin{equation}
        K_\alpha[\mathcal{D}_\phi]=d^3\e{(1-2\alpha)M_{2\alpha}(\phi)}-d^2.
    \end{equation} \end{proof}
Now we can prove \cref{thm2}.  
\begin{proof}
    Notice that \cref{lemm2} works for any choice of SIC-POVM, group covariant or not. In particular, it will be true if we consider a WH-SIC generated as an orbit of a Weyl-Heisenberg group. Moreover, the quantity $K_\alpha[\mathcal{V}]$ is well defined for any real number $\alpha\geq1$. Hence, recalling that the stabilizer entropies are defined for $\alpha\geq2$, \cref{lemm2} and \cref{lemm} taken together imply the following inequality valid for any state $\ket\phi$:
    \begin{align}
        M_\alpha(\phi)&=\frac{1}{1-\alpha}\log{\frac{1}{d^3}\left(K_{\alpha/2}[\mathcal{D}_\phi]+d^2\right)}\\&\leq\frac{1}{1-\alpha}\log{\frac{1+(d-1)(d+1)^{1-\alpha}}{d}}
    \end{align}
    with equality if and only if the state $\ket{\phi}$ is a WH-SIC fiducial state.
\end{proof}

\section{Discussion} \label{5}
In this work we have fully characterized the states achieving the upper bound on the stabilizer entropies in \cref{upb}:  They must belong to a Weyl-Heisenberg covariant SIC\@. Our result establishes a deep connection between the SICs and the stabilizer states, putting them at two extremes of the spectrum of magic. 

This is interesting for a number of reasons. First, it establishes the WH-SIC states as the \emph{uniquely} most magical states in the Hilbert space, supposing they do indeed exist.  This confirms a suspicion that has been circulating in the community since at least 2014~\cite{andersson_states_2015}. Moreover, it relates the problem of finding the maximum values of magic to the historically hard problem of SIC existence. 

In this regard, there is quite a deep lesson here.  The problem of WH-SIC construction in any $d\ge4$ is known to be deeply number theoretical in its essence~\cite{Bengtsson2016,Bengtsson2020}. (See Refs.\ \cite{Appleby2021,Bengtsson2024,Appleby2025} for some of the latest developments.)  In fact it has been speculated~\cite{FuchsStaceyNSF} that a quantum computer may be needed for calculating SICs efficiently because of the difficulty of computing the unit groups of arbitrary algebraic number fields~\cite{Eisentraeger2014}.  Our result shows that magic, at least as quantified by the stabilizer entropies, will inherit exactly the same difficulties.

We may comment that this observation will also be true for another quantifier of magic~\cite{Feng_2025} that we only become aware of while writing this paper---this time a concave function, in contrast to the convex stabilizer entropies considered here.  That quantifier is based on the Zhu's concept of a ``1/2-design''~\cite{Zhu2022}, which requires separate techniques to elucidate.

On this note, it is insightful that while our upper bound in \cref{thm2} was established independently of the choice of the WH group---single qudit vs compound subsystems---this is clearly not true about the existence of potential SICs associated with them. In fact, SICs do not play so nicely with products of Weyl-Heisenberg groups. In \cite{Godsil_2009}, it was proven that the existence of SICs covariant under a tensor product of $n$ copies of the single-qubit Pauli group (that is $\mathcal{W}(2)$) is not allowed, unless $n=3$. Therefore, our result confirms that the upper bound of the stabilizer entropy is not achievable in the case of compound WH groups, except when $d=2^3$. 

Related to this, if one still considers different WH groups in Hilbert spaces of the same dimension, as e.g., $\mathcal{W}(4)$ vs.\ $\mathcal{W}(2)^{\otimes2}$, one sees that it is impossible to reach our bound when treated as a bi-partite structure, though it possible for the single global WH group~\cite{Renes_2004}. This example hints at the fact that local non-Clifford operations that act only on single subsystems are limited in the amount of magic they can inject into a system. It also makes explicit how strong the dependence of the Clifford structure is on the structure of the system's Hilbert space, i.e., whether treated as a tensor product or not. 

Finally, let us mention something of a conceptual puzzle brought up by our result.  SICs have been of interest to the quantum foundations program of QBism for a number of years because they help draw out the analogy between the quantum mechanical Born rule and the (classical) law of total probability~\cite{Fuchs2013,Appleby2017}.  In fact, with their aid, one can bring the Born rule into a form that is as close as possible to the classical law~\cite{DeBrota2020}.  So, it is a bit surprising that the key feature of SICs that made those results true is, in fact, the same feature that powered our result here that the SICs would be a maximally nonclassical resource.  The key feature is that the SICs play the role of quasi-orthonormal bases in the cone of positive semi-definite operators~\cite{appleby_symmetric_2010}. In that sense, SICs are as close to classical as can be.  This contrast is also manifested in  another result of one of us~\cite{Fuchs2004,Audenaert2004} that such states would be the most sensitive to eavesdropping in a measure and re-prepare strategy---again a maximally nonclassical resource.  We simply comment that this duality is interesting and perhaps meaningful.

\section{Acknowledgements}
GC thanks G. Esposito for advice and insightful discussions. This work was supported in part by National Science Foundation Grants 2210495 and 2328774, and in part by Grant 62424 from the John Templeton Foundation. The opinions expressed in this publication are those of the authors and do not necessarily reflect the views of the John Templeton Foundation.

\bibliographystyle{utphys3}
\bibliography{Bib}

\end{document}